\documentclass[11pt]{article}
\usepackage{verbatim}
\usepackage{amsthm}
\usepackage[active]{srcltx}
\usepackage[colorlinks,linkcolor={blue},
citecolor={blue},urlcolor={red},]{hyperref}
\usepackage{hyperref}
\usepackage{setspace}
\usepackage{amsfonts}
\usepackage{amsmath, amsthm}
\usepackage{anysize}
\usepackage{enumerate}
\usepackage{proof}
\usepackage{graphicx}
\usepackage{epstopdf}
\usepackage[round]{natbib}
\usepackage{bibentry}
\usepackage{pdfpages}

\nobibliography*

 {\begin{list}{}%
         {\setlength{\leftmargin}{#1}}%
         \item[]%
 }
{\end{list}}

\newcommand{\ds}{\displaystyle}
\newcommand{\fr}{\frac}
\newcommand{\pr}{\partial}

\newcommand{\Ex}{\mathbb{E}}
\newcommand{\mQ}{\mathbb{Q}}
\newcommand{\mP}{\mathbb{P}}

\newcommand{\filt}[1]{\mathcal{F}_{#1}}

\newcommand{\VIX}{V}
\newcommand{\iVIX}{\widetilde \VIX}
\newcommand{\vF}{\hat F^{\VIX}}
\newcommand{\ExQ}[1]{\Ex^\mQ\left[\left.{#1}\right|\filt{t}\right]}

\newcommand{\R}{\mathbf{R}}

\newcommand{\Rp}{\R_+}

\newcommand{\qv}[2]{\left\langle{#1},{#2}\right\rangle}

\newcommand{\vX}{{\bf{X}}}

\newcommand{\vx}{{\bf{x}}}

\newtheorem{thm}{Theorem}[section]

\newtheorem{prop}[thm]{Proposition}
\newtheorem{lem}[thm]{Lemma}
\newtheorem{cor}[thm]{Corollary}

\theoremstyle{definition}

\newtheorem{defn}[thm]{Definition}
\newtheorem{condition}[thm]{Condition}
\newtheorem{assumption}[thm]{Assumption}

\theoremstyle{remark}
\newtheorem{remark}{Remark}[section]
{
   \end{minipage}
   \vspace*{\stretch{3}}
   \clearpage
}

\begin{document}
\title{A Market Model for VIX Futures}
\author{Alexander Badran and Beniamin Goldys}
\date{\today}
\maketitle

\begin{abstract}


A new modelling approach that directly prescribes dynamics to the term structure of VIX futures is proposed in this paper.  The approach is motivated by the tractability enjoyed by models that directly prescribe dynamics to the VIX, practices observed in interest-rate modelling, and the desire to develop a platform to better understand VIX option implied volatilities.  The main contribution of the paper is the derivation of necessary conditions for there to be no arbitrage between the joint market of VIX and equity derivatives.  The arbitrage conditions are analogous to the well-known HJM drift restrictions in interest-rate modelling.  The restrictions also address a fundamental open problem related to an existing modelling approach, in which the dynamics of the VIX are specified directly. The paper is concluded with an application of the main result, which demonstrates that when modelling VIX futures directly, the drift and diffusion of the corresponding stochastic volatility model must be restricted to preclude arbitrage.  

%
\end{abstract}
\emph{Keywords: VIX, VIX Futures, VIX Options, Market Model, HJM}
\tableofcontents
\newpage
\label{chpt:VFM}
\section{Introduction}
The Chicago Board of Options Exchange Volatility Index, which is more commonly known as the VIX, is an approximation of the markets' expectation of index volatility over a 30-day time period.
Derivatives on the VIX provide market participants with a mechanism to invest in the markets' expectation of volatility, without the need for purchasing index options.  Investors can gain exposure to volatility through the purchase of VIX futures or VIX options, which have become increasingly popular in recent years. In 2004, futures on the VIX began trading and were subsequently followed by options on the VIX in 2006.  Since products are traded on both the underlying index and the VIX, it is desirable to employ a model that can simultaneously reproduce the observed characteristics of products on both indices, while remaining free from arbitrage.

In this paper, a new modelling approach that directly prescribes dynamics to the term structure of VIX futures is proposed.  The approach is motivated by the tractability enjoyed by models that directly prescribe dynamics to the VIX, practices observed in interest-rate modelling and the desire to develop a platform to better understand VIX option implied volatilities. In the existing literature, many solutions have been proposed for the joint-modelling task.  The existing approaches can be categorised according to the assumptions made regarding the market of traded instruments.  All existing models can generally be placed into one of the following categories:\\
\indent (a) The underlying index is the primary traded security.  Index options, VIX futures and VIX options are priced relative to the underlying index.\\
\indent (b) The underlying index and a continuum of variance swaps are the primary traded securities.  Index options, VIX futures and VIX options are priced relative to these products.\\
\indent (c) The VIX is modelled directly.  VIX futures and VIX options are priced relative to the VIX.

For models that belong to category (a), the dynamics for the underlying index are specified under a pricing measure and the discounted price of a derivative is expressed as a local martingale.  The square of the VIX is defined as the expected realised variance of the index and the discounted price of a derivative on the VIX is expressed as a local martingale under the same measure.

The majority of the literature on VIX derivatives fall into this category.  \cite{ZhangZhu06} derived an expression for VIX futures assuming \cite{Heston93} stochastic volatility dynamics.  \cite{Lin07} presented an approximation formula for VIX futures based on a convexity correction, which was then used to price VIX futures when the S\&P500 is modelled by a Heston diffusion process with simultaneous jumps in the underlying index and the volatility process (SVJJ).  A more general result was presented in \cite{ZhuLian11}, who assumed the same dynamics for the S\&P500 as \cite{Lin07} and derived an exact formula to price a VIX futures contract.   The literature on VIX options is generally similar to that of VIX futures: affine stochastic-volatility dynamics are usually assumed for the underlying index, which enable some tractability to be retained when deriving option pricing formulae. A square-root stochastic variance model with variance jumps and time-dependent parameters was considered for the evolution of the S\&P500 index in \cite{Sepp08b}, while option pricing formulae under the SVJJ dynamics were presented in \cite{LianZhu11}.  \cite{Sepp11} and \cite{PaSi13} attempted to capture empirically observed features of the implied skew for options on the VIX.  The former considered a range of parametric and non-parametric models, while the latter employed a regime-switching stochastic-volatility model.  A non-affine 3/2 plus jumps model was considered in \cite{BB}.  Pricing formulae were derived for VIX derivatives and numerical results illustrated that the 3/2 model is a good candidate for the joint modelling of VIX and equity derivatives.

For models belonging to the second approach, it is assumed that variance swap markets are liquid enough to justify the use of variance swaps as a model input.  The extended market enables models to better capture term-structure features observed in both variance and volatility markets.  Similar to the previous class of models, the VIX is defined in terms of the primary instruments and the discounted price of a derivative on the VIX is expressed as a local martingale under the same pricing measure that is used for pricing derivatives on the index.  See \cite{Bergomi05}, \cite{Bergomi08}, \cite{Buehler06} and \cite{Gath08} for a further discussion of models in this class.

A first step in proposing any model is to make assumptions regarding the class of traded instruments.  In approaches (a) and (b), it is assumed that index instruments and possibly variance swaps are liquidly traded.  A model is assumed for the dynamics of the liquidly traded instruments, which enables an expression for the VIX to be derived.  By construction, these approaches guarantee that the assumed dynamics for the VIX are consistent with those assumed for the underlying securities.  Hedging VIX derivatives in this framework, however, is a non-trivial exercise.  Unlike other volatility related products, the VIX is not traded and it cannot be statically replicated, due to the non-linear transformation used in it's definition, making hedging in practice a non-trivial exercise in the setup of approaches (a) and (b).

A related modelling approach to the one proposed in this paper is that of category (c), in which the dynamics of the VIX are specified directly. Dynamics are assumed under a pricing measure and derivatives are priced as discounted expectations of their future payoffs.
There are several examples of this approach in the literature.  \cite{GrLo96} considered a mean-reverting square-root process for the evolution of the VIX and presented closed-form pricing formulae for VIX derivatives. \cite{PsDM10} concluded that a mean-reverting logarithmic diffusion with jumps is supported by VIX time-series data and also derived closed-form formulae for VIX derivatives. A variety of model specifications were considered by \cite{KaAl10}.  The authors evaluated the performance of a wide range of models for risk management and derivatives pricing applications. An empirical analysis of one-dimensional diffusions for the VIX was performed in \cite{GoardMaz12} and the authors concluded that a pure-diffusion 3/2 model is best suited to capture the dynamics of the VIX.  Derivatives were then priced relative to the VIX under such a specification.  \cite{DF12} attempted to replicate the concept of a local-volatility surface, which originated in \cite{Dupi93}, for VIX derivatives under the assumption of linear mean-reverting dynamics.  The authors justify ignoring the dynamics of the underlying index with the claim that the VIX market is mature enough for the pricing and hedging of VIX options relative to VIX futures, which is the market practice.

There are several complexities associated with these models.  The first challenge is in ensuring that the joint market between the underlying index and derivatives on the VIX is free from arbitrage.  To ensure that the markets are arbitrage free requires the derivation of restrictions on the dynamics of the VIX, which is a problem that has not been addressed in the literature.  The derivation of restrictions that ensure no dynamic arbitrage is a well-known problem in other areas of finance.  In interest-rate modelling, the HJM drift conditions (\cite{HeJM92}) ensure that there is no arbitrage when forward rates are modelled directly.  Variance curve models for variance swap markets are analogous to forward-rate models for interest-rate markets. \cite{Buehler06} derived variance curve arbitrage conditions as well as addressing the problems of finite-dimensional realisations and model consistency. Many attempts have been made to produce similar results for option market models by directly prescribing dynamics to Black-Scholes implied volatilities (see for example \cite{Scho99}, \cite{BGKW01} and \cite{ScWi08}).  The situation is much more complex for options, however, due to the higher dimensionality of the state space and the non-linearity of Black-Scholes implied volatilities.  The main contributions of this paper are Theorem~\ref{thm:VFM:EMM1} and Theorem~\ref{thm:VFM:EMM2}, which state necessary conditions for there to be no arbitrage between the joint market of VIX and equity derivatives. 

Another complexity associated with this modelling approach is in the appropriate specification of a {market price of risk}.  Since the VIX is not tradable and cannot be replicated, the usual relationships that connect a derivative to its underlying are not typically observed.  VIX futures are not restricted by traditional cost-of-carry relationships and VIX options violate put-call-parity relationships when compared to the spot. 
By modelling VIX futures directly, as opposed to the VIX itself, complexities involved with the appropriate choice of the market price of risk are avoided.  This is similar to the comparison of short-rate models to forward-curve models in interest-rate modelling.

The final motivating factor for modelling VIX futures, as opposed to the VIX directly, is the concept of VIX option implied volatility.  To properly understand mathematical features of VIX option implied volatilities, a framework that connects the dynamics of the VIX futures to the underlying index is required.  \cite{CoHo05} and \cite{Rope10} provide a comprehensive discussion of no-arbitrage restrictions for traditional option and implied volatility surfaces. To do so, the authors relied on the notion of an equivalent martingale measure.  In order to extend these concepts to the joint market of VIX and equity derivatives, necessary conditions for the existence of such a measure are required. 

The remainder of this paper is set out as follows.  Market conventions, the modelling framework and mathematical preliminaries for the paper are specified in Section~\ref{sect:VFM:Preliminaries}.  In Section~\ref{sect:VFM:DerivedDynamics}, a general semi-martingale representation for the VIX index is derived.  The dynamics are stated in Proposition~\ref{prop:VIXSDE:VIXSDE}.  The representation is quite tedious and an alternate form of the diffusion term in a simplified setting is provided in Corollary~\ref{cor:VMM:diffusion}.  Section~\ref{sect:VFM:ImpliedDynamics} is concerned with the implications of modelling the term structure of VIX futures directly.  The dynamics of a process that represents the VIX as implied by the family of equations for the VIX futures is derived in Proposition~\ref{prop:VFM:VIXdiffusion}.
Section~\ref{sect:VFM:Consistency} is concerned with the implications of the joint modelling of the underlying index and the term structure of VIX futures.  Restrictions on the dynamics stated in Section~\ref{sect:VFM:DerivedDynamics} and Section~\ref{sect:VFM:ImpliedDynamics} are derived so that there is no arbitrage between the joint market of VIX and equity derivatives.  The restrictions are stated in Theorem~\ref{thm:VFM:EMM1}.  The paper is concluded with Section~\ref{sect:VFM:Applications}, where an application of the main theorem is provided. The application demonstrates that by modelling the VIX futures directly, the drift and diffusion of the corresponding stochastic volatility model must be restricted to preclude arbitrage.

\section{Preliminaries}
\label{sect:VFM:Preliminaries}
Market conventions and the modelling framework adopted throughout the paper are described in this section.  The market definition of the VIX is first stated, followed by the theoretical definition of the VIX and the assumptions that are made regarding the dynamics of the traded instrument.
\subsection{Market Conventions}
%
    The VIX attempts to provide investors with a model-free measure of the markets' expectation of 30-day volatility of the S\&P500 index.
    The market definition of the VIX is based on a representation of expected realised variance in terms of option contracts and is stated below.

\begin{defn}
\label{defn:Defn:MarketDef}
    The \textit{CBOE VIX} is calculated using the formula
         \begin{equation}
                		\widetilde{VIX}_t^{mrk}(T) := \sqrt{\left[\fr{2e^{r(T-t)}}{T-t}\sum_{j\in I} \Theta_j\fr{\Delta k_j}{k_j^2}-\fr{1}{T-t}\left(\fr{F_{PCP}}{k_0}-1\right)^2\right]}\times 100;\label{eqn:Defn:CBOEVIX}
        \end{equation}
            where
            \begin{itemize}
            		\item[$T-t$:] is the time horizon (typically 30 days).
            		\item[$r$:] is the risk-free interest rate that applies from time $t$ till time $T$.
            		\item[$k_0$:] is the ``at-the-money" strike and is given by the strike that minimizes the difference between put and call option prices at time $t$ with expiry at time $T$.
            		\item[$k_j$:] is the $j^{th}$ strike and $k_j<k_{j+1}$.
            		\item[$\Theta_j$:] is the price of an out-of-the-money option at time $t$ with strike $k_j$ and expiry at time $T$, computed as the average of the bid-ask spread. If $j<0$ put options are used, if $j>0$ call options are used and for $j=0$ the average of the put and call price is used.
            		\item[$I$:] is the set of all $j$, ordered by strike, for which quoted strikes exist in the market with the provisos that:
            					\begin{itemize}
            							\item if the bid price for $\Theta_j$ is zero, the $j\notin I$;
            							\item the summation stops if two consecutive zero bid prices are met.
            					\end{itemize}
            		\item [$\Delta k_j$:] is the central symmetric difference $\fr{1}{2}(k_{j+1}-k_{j-1})$, except for the first and last strikes in the sum where a one-sided difference, whichever of $k_j-k_{j-1}$ or $k_{j+1}-k_j$ is appropriate, is used.
            		\item [$F_{PCP}$:] is the forward index level, computed using \[F_{PCP}:=k_0+e^{r(T-t)}(C_0-P_0),\]
            					where $C_0$ and $P_0$ are the mid-prices of the call and the put options at time $t$ with strike $k_0$ and expiry at time $T$.
            \end{itemize}
In the event that no options have exactly 30 days till expiration, as is ordinarily the case,
the CBOE interpolates between the CBOE VIX squared calculated at the two closest maturities, that is,
\begin{equation}
		VIX^{mrk}_t := \sqrt{\fr{365}{30}\left((T_1-t)\widetilde{VIX}_t^2(T_1)\fr{N_{T_2}-30}{N_{T_2}-N_{T_1}}\right.
                        \left.+(T_2-t)\widetilde{VIX}_t^2(T_2)\fr{30-N_{T_1}}{N_{T_2}-N_{T_1}}\right)}
                         \label{eqn:Defn:PortfolioVIX}
\end{equation}
where $t<T_1<t+30\mbox{ days}<T_2$ and $N_{T}$ denotes the number of days from $t$ till time $T$. Since the CBOE VIX is always defined for a 30-day time horizon, the dependence on $T$ is omitted.
\end{defn}
%

The VIX is defined using liquidly quoted index options, however, the index itself is not tradable.  Investors are able to take a position on the value of the VIX index via VIX futures and European options on the VIX, which are liquidly traded instruments. VIX futures and VIX options deliver a cash settlement amount that is related to the value of the VIX index at expiry.  Traditional relationships that are usually observed between an option and its underlying do not hold for VIX options, as a consequence of the fact that the VIX index is not tradable.  The relationships, however, are observed between VIX options and the VIX futures.

For more information on the market definition of the VIX and on VIX derivatives, the reader is referred to the \cite{CBOE09} White Paper.
Since the definition of the VIX is rather cumbersome, simplifications are adopted in the literature to enable mathematical tractability.  In the following section, the definitions adopted in this paper and the simplifications often observed in the literature are presented.

\label{sect:VFM:Setup}
\subsection{Theoretical Definition of the VIX}

In this paper, the VIX is defined in terms of European options.  This is done to avoid a potential loss in generality, since there are no implicit model assumptions associated with such a definition.  The definition is independent of any assumptions regarding a risk-neutral measure and absence of arbitrage.  Upon making additional assumptions, a representation in terms of expected realised variance may be recovered.  Until such assumptions are made, however, the representation may not hold.  This is subtle point and it is illustrated in the forthcoming Lemma~\ref{cor:3on2:VIXasLog}.
\begin{remark}
The convention throughout is to use a circumflex to denote processes that are in terms of a fixed expiration date, while no accent is used to denote processes that are in terms of time-till expiry.
\end{remark}

\begin{remark}
     For the following definition, further assumptions are required to ensure that $\hat\VIX_t$ is well behaved.  For example, out-of-the-money put options need to decay fast enough as $k\searrow0$ to ensure that $\hat\VIX_t$ does not explode.  Additional assumptions are generally made to ensure that $\hat\VIX_t$ is well defined.
\end{remark}

\begin{defn}
    \label{defn:Defn:VIXContKOptionsFixT}
     For all $T>t$, let
     \begin{equation}
              \hat\VIX_t(T) :=\sqrt{\fr{1}{T-t}\int_0^\infty\fr{\hat\Theta_t(k,T;\hat F_t(T))}{k^2}\,m(dk)},\quad\forall\, 0\,\le t<T,
    \end{equation}
    where $\hat F_t(T)$ is an index future at time $t$ expiring at time $T$, $\hat\Theta_t(k,T;\hat F_t(T))$ is an out-of-the-money index option with strike $k$ and maturity $T$, and $m(\cdot)$ is an unspecified measure.
\end{defn}
\begin{remark}
    The behaviour of $\hat\VIX_t(T)$ as $T\searrow t$ is dependent on the behaviour of $\hat\Theta_t(\cdot,T;\hat F_t(T))$ as $T\searrow t$.  While being of mathematical importance, this limiting behaviour is not investigated in this paper, since the VIX is defined for fixed $T>t$.
\end{remark}

\begin{defn}
    \label{defn:Defn:VIXContKOptions}
    For all $t\ge0$ and for fixed $\tau^*>0$, the \textit{VIX} is given by
          \begin{equation}
                      \VIX_t:=\hat\VIX_t(t+\tau^*),\quad\forall t>0. \label{eqn:Defn:VIXContKOptions}
          \end{equation}
    Since $\tau^*$ is a fixed constant throughout, there is no loss in generality incurred by suppressing it as an argument to the process $V$. 
    For simplicity, the dependence of the futures contract on $t+\tau^*$ is suppressed in the forthcoming analysis.
\end{defn}
A different choice of $m(\cdot)$ can be made depending on the specific application.  In particular, to recover a discrete-strike definition of the VIX, one would choose some form of discrete measure, whereas for a continuous-strike setup, one would chose the Lebesgue measure. The following definition is motivated by the continuous-strike definition of the VIX, which is a common simplification observed in the literature.

\begin{defn}\label{def:Intro:CS}
    Suppose that the measure $m(\cdot)$ is chosen as
    \begin{equation*}
        m(\cdot) = N\times m_L(\cdot),
    \end{equation*}
    where $N\equiv2\times100^2$ and $m_L(\cdot)$ denotes the Lebesgue measure.  Then the VIX is given by
    \begin{equation*}
        \VIX_t=\sqrt{\fr{N}{\tau^*}\int_0^\infty\fr{\Theta_t(k,t+\tau^*; F_t(t+\tau^*))}{k^2}\,dk}
    \end{equation*}
    and this quantity is referred to as the \textit{continuous-strike} definition of the VIX.  
\end{defn}

Under the specification given in Definition~\ref{def:Intro:CS} and upon making addition assumptions regarding the underlying index, a simpler representation for the VIX can be derived.  Such a representation is the typical starting point for most models and is the subject of the following lemma.
\begin{lem}
\label{cor:3on2:VIXasLog}
Suppose that the measure $m(\cdot)$ is chosen as specified in Definition~\ref{def:Intro:CS}.  Further suppose that there exists a measure $\mQ$, such that index futures and options are $\mQ$-martingales.  Then the VIX can be written as
\begin{equation}
\VIX_t = \sqrt{- \fr{N}{\tau^*} \ExQ{\ln \left( \frac{F_{t+\tau^*}}{F_t} \right)}}.
\end{equation}
\begin{proof}

    For fixed $t\ge 0$, fix $\omega\in\Omega$.  Since $F_{t+\tau^*}(\omega)\in\R_+,\,\,\forall\omega\in\Omega$, for $f(s)=\ln(s)$ and $X=F_{t+\tau^*}(\omega)$, Lemma~\ref{lem:Defn:BR} yields
\begin{align*}
    \ln(F_{t+\tau^*}(\omega))={}&\ln(F_t)+\fr{1}{E}(F_{t+\tau^*}(\omega)-E)\\
            {}&-\int_E^\infty\fr{(F_{t+\tau^*}(\omega)-k)^+}{k^2}\,dk-\int_0^E\fr{(k-F_{t+\tau^*}(\omega))^+}{k^2}\,dk.
\end{align*}
Since $\omega$ is arbitrary, the above equation is true for any $\omega\in\Omega$.  Taking conditional expectations with respect to $\filt{t}$ under the measure $\mQ$ and rearranging gives
\begin{align}
  \int_E^\infty\fr{1}{k^2}\,C_t(k,t+\tau^*;F_t)\,dk+{}& \int_0^E\fr{1}{k^2}\,P(k,t+\tau^*;F_t)\,dk \nonumber \\{}&=-\ExQ{\ln(F_{t+\tau^*})}+\ln(F_t)+\fr{1}{E}(F_{t}-E).\label{eqn:Intro:proofVIXasLog}
\end{align}
Evaluating the above equation for $E=F_t$ and multiplying both sides by $2\times100^2/\tau^*$ completes the proof.
\end{proof}
\end{lem}

 The convention is to approximate the square of the VIX by the risk-neutral expectation of the realised variance of the underlying index. The VIX at time $t$ is often approximated by
    \begin{equation}
       \VIX_t\approx\left. -\fr{100}{\sqrt{\tau^*}}\times\sqrt{\Ex^\mQ\left(\left.[\ln F]_{t+\tau}-[\ln F]_t\right|\filt{t}\right)}\right|_{\tau=30\mbox{ days}},
    \end{equation}
where $[\cdot]$ is used to denote quadratic variation and $F_t$ denotes an index future expiring at time $t+\tau$.  There are many different representations of the expected realised variance of a stochastic processes and one of these representations, such as the logarithmic representation presented in Lemma~\ref{cor:3on2:VIXasLog}, is typically chosen for the definition of the VIX.  A consequence of defining the VIX in terms of expected realised variance is that model assumptions are implicit in the definition.  These model assumptions are fundamental to the origins of the VIX, however, {the market definition of the VIX at any fixed point in time is simply a function of observed market prices and is independent of any previously made assumptions}.
Lemma~\ref{cor:3on2:VIXasLog} may not hold, for example, if the index is a strict-local martingale as opposed to a true martingale.
\subsection{Index and VIX Futures Dynamics}

The modelling framework and the assumptions made regarding the class of traded instruments are stated in this section.  The setup is rather general; the assets are assumed to be continuous processes that admit a finite-dimensional realisation of arbitrary size.
Consider a continuous-time economy with trading interval $[0,T^*]$ for a fixed horizon date $T^*>0$.
Let $(W_t=(W_t^0,W_t^1,...,W_t^d))_{0\le t\le T^*}$ be a $(d+1)$-dimensional Brownian motion on $(\Omega,\mathcal{F},\mathbb{F},\mathbb{P})$, where $\mathbb{F}=(\filt{t})_{0\le t\le T^*}$ is the $\mathbb{P}$-augmented filtration generated by $W$.

\begin{assumption}
\textbf{(A1)}
For each $i=0,...,d$, let $\sigma^i:[0,T^*]\times\Omega\to\Rp^{d+1}$ and $\mu^i:[0,T^*]\times\Omega\to\R$ be $\mathbb{F}$-adapted processes, with \[\int_0^{T^*}|\mu^i_s|\,ds<\infty\quad\mbox{and}\quad\int_0^{T^*}|\sigma^i_s|^2\,ds<\infty,\quad\mP\mbox{-a.s..}\]
The vector of processes $(\vX_t=(X_t^0,X_t^1,...,X_t^d))_{0\le t\le T^*}\in \R^{d+1}$ is assumed to satisfy
\begin{align}
  X^i_t &=X^i_0+\int_0^t X^i_s\,\mu^i_s \,ds +\int_0^tX^i_s\,\sigma^i_s\cdot dW^\mP_s,\label{eqn:VFM:Xvect}
\end{align}
for each $i=0,...,d$.  Let $F\equiv X^0$, so that the underlying index future $(F_t=X_t^0)_{0\le t\le T^*}$ is assumed to be a stochastic process with dynamics given by
\begin{align}
  F_t &= F_0 + \int_0^t F_s\,\mu^0_s \,ds +\int_0^tF_s\,\sigma^0_s \cdot dW^\mP_s. \label{eqn:VFM:F_SDE}
\end{align}
\end{assumption}

The motivation for such a setup is to allow for existing models, such as stochastic volatility models, to be examined within the proposed framework.

\begin{assumption}
\textbf{(A2)}
For all times $t\in[0,T^*]$, the market contains call and put options for all strikes $k\in[0,\infty)$ and time-till expiries $\tau\in(0,T^*-t]$.
The price of an option at time $t\in[0,T^*]$ with strike $k\in[0,\infty)$ and expiry $\tau\in(0,T^*-t]$ is derived from $\vX_t$ and is assumed to be a deterministic function, denoted by $C(k,\tau,t,\vX_t)$ for a call and $P(k,\tau,t,\vX_t)$ for a put.
\end{assumption}
%
%

For all $0\le t\le T\le T^*$, let $\vF(t,T)$ denote the value of a VIX future at time $t$ expiring at time $T$.  The following assumption is concerned with the dynamics of the VIX futures.  The purpose of the following assumption is to allow for a very general specification for the VIX futures and no interpretation or context for the measure $\mP$ is initially provided.
\begin{assumption}
\textbf{(A3)}
For all $0\le T\le T^*$, let $\nu(\cdot,T):[0,T]\times
\Omega\to \Rp^{d+1}$ and $\mu^\VIX(\cdot,T):[0,T]\times \Omega\to \R$ be $\mathbb{F}$-adapted processes.
For all $0\le t<T\le T^*$, the family of equations for the VIX futures is assumed to satisfy
\begin{align}
    \vF(t,T) ={}& \vF(0,T)+\int_0^t\vF(u,T)\,\mu^\VIX(u,T)\,du
    +\int_0^t\vF(u,T)\,\nu(u,T)\cdot dW^{\mP}_u.\label{eqn:VFM:F}
\end{align}

\begin{remark}
For results on the the existence and uniqueness of solutions to equations of the form of Equation~\eqref{eqn:VFM:F}, the reader is referred to Section 4.6 of \cite{Mort88}.
\end{remark}

Further assume that, for any $0\le T\le T^*$,
    \begin{equation*}
        \int_0^T|\mu^\VIX(u,T)|\,du+\int_0^T|\nu(u,T)|^2\,du<\infty,\quad \mP\mbox{-a.s.},
    \end{equation*}
and that the limit
\begin{align*}
    \vF(t,t)&=\lim_{T\searrow t}\vF(t,T)\\
        &= \vF(0,t)+\int_0^t\vF(u,t)\,\mu^\VIX(u,t)\,du
    +\int_0^t\vF(u,t)\,\nu(u,t)\cdot dW^{\mP}_u,
\end{align*}
is well defined for all $0\le t<T^*$, $\mP\mbox{-a.s.}$.
\end{assumption}

Given the above family of equations, one can introduce the following process that represents the VIX as implied by the VIX futures.
\begin{defn}
    The \textit{implied VIX} is given by
\begin{equation*}
    \iVIX_t:=\vF(t,t),
\end{equation*}
for all $t\in[0,T^*]$.
\end{defn}
Section~\ref{sect:VFM:DerivedDynamics} is concerned with the derivation of the dynamics of the VIX from the underlying index and the setup specified in (A1) and (A2) is assumed.  In Section~\ref{sect:VFM:ImpliedDynamics}, the setup of (A3) is all that is assumed. The implied dynamics of the VIX is the focus of this section and the analysis is independent of any specification for the underlying index.  The focus of Section~\ref{sect:VFM:Consistency} and Section~\ref{sect:VFM:Applications} is on the complete framework of (A1)~-~(A3).  

\section{Deriving the Dynamics of the VIX from the Index}
\label{sect:VFM:DerivedDynamics}
In this section,
a general semi-martingale representation of the VIX is derived under the measure~$\mP$, which is used to denote the real-world or statistical measure. A typical starting point for most market models is to directly specify the dynamics of the modelled quantities under an equivalent martingale measure. This is done to avoid the complexities involved with the specification of a market price of risk and with a change in measure. The reason for performing an analysis of the real-world dynamics is that the VIX is often used in empirical investigations, due to its role as an indicator for market sentiment.  Starting under the real-world measure also avoids complexities regarding the existence of a risk-neutral measure. The ultimate goal, however, is to provide a framework for the pricing and hedging of derivatives, which is typically done under an equivalent risk-neutral measure.  The existence of such a measure is not discussed in this section, rather, it is the subject of Theorem~\ref{thm:VFM:EMM1} and Theorem~\ref{thm:VFM:EMM2} in the forthcoming Section~\ref{sect:VFM:Consistency}.

This section is structured as follows.  The Ito-Ventzel formula is first applied to derive the dynamics of the square of the VIX when defined in terms of options with fixed expiry.  A Musiela-like parameterisation in terms of fixed time-till maturity is then introduced, which allows for the derivation of a governing stochastic differential equation for the VIX.  The representation is presented in Proposition~\ref{prop:VIXSDE:VIXSDE}.   In Corollary~\ref{cor:VFM:IndexToVIX}, the dynamics implied by the specification in Section~\ref{sect:VFM:Setup} are stated.
%
%
%
\begin{prop}
\label{prop:VFM:Vhat2}
Suppose that the setup of (A1) and (A2) in Section~\ref{sect:VFM:Setup} are assumed, that is, at time $t$ puts and calls on the index with strike $k$ and expiry at time $T$ are given by the functions $\hat P(k,T,t,\vX_t)$ and $\hat C(k,T,t,\vX_t)$.  Further assume that, for all $T\in (0,T^*]$,
\[\hat P(\cdot,T,\cdot,\cdot),\, \hat C(\cdot,T,\cdot,\cdot)\in\mathcal{C} \quad\mbox{and}\quad\hat P(\cdot,T,\cdot,\cdot),\, \hat C(\cdot,T,\cdot,\cdot)\in \mathcal{C}^{1,1,2}\]
on $\Rp\times [0,T)\times\Rp^{d+1}$, for all $(k,t,\vx)\in\Rp\times [0,T^*)\times\Rp^{d+1}$,
\[\hat P(k,\cdot,t,\vx),\,\hat C(k,\cdot,t,\vx)\in \mathcal{C}^1\] on $(t,T^*]$,
 for all $\vx\in\R^{d+1}$ and $0\le t<T\le T^*$,
\begin{align*}
    \int_0^\infty\fr{1}{k^2}\hat\Theta(k,T,t,\vx)\,m(dk)<\infty,\quad
    \int_0^\infty\fr{1}{k^2}\left|\fr{ \pr\hat\Theta}{\pr t}(k,T,t,\vx)\right|\,m(dk)<\infty,
\end{align*}
and for all $\vx\in\R^{d+1}$, $0\le t<T\le T^*$ and $i,j=0,...,d,$
\begin{align*}
    \int_0^\infty\fr{1}{k^2}\left|\fr{\pr\hat\Theta}{\pr x_i}(k,T,t,\vx)\right|\,m(dk)<\infty,\quad\mbox{and}\quad
    \int_0^\infty\fr{1}{k^2}\left|\fr{\pr^2\hat\Theta}{\pr x_i\pr x_j}(k,T,t,\vx)\right|\,m(dk)<\infty.
\end{align*}
Then, for all $T\in (0,T^*]$ and $t\in[0,T)$, the dynamics of the process $\hat\VIX^2_t(T)$ are given by
\begin{align}
    d\hat\VIX_t^2(T)={}&-\fr{1}{T-t}\hat\VIX_t^2(T)\,dt+\fr{1}{T-t}\int_0^\infty\fr{1}{k^2}\fr{\pr\hat \Theta}{\pr t}(k,T,t,\vX_t)\,m(dk)\,dt\nonumber\\
            {}&+\fr{1}{T-t}\sum_{i=0}^d\int_0^\infty\fr{1}{k^2}\fr{\pr\hat \Theta}{\pr x_i}(k,T,t,\vX_t)\,m(dk)\,dX^i_t-\fr{1}{2(T-t)}\fr{1}{F_t^2}d\qv{F}{F}_t\nonumber\\
                {}&+\fr{1}{2(T-t)}\sum_{i,j=0}^d\int_0^\infty\fr{1}{k^2}\fr{\pr^2\hat \Theta}{\pr x_i\pr x_j}(k,T,t,\vX_t)\,m(dk)\,d\qv{X^j}{X^i}_t.\label{eqn:DR:wtVIXsquared}
\end{align}
%

\begin{proof}
    See Appendix~\ref{app:proof:propVhat2}.
\end{proof}
\end{prop}
\begin{remark}
    The quadratic variation term and the lack of symmetry in Equation~\eqref{eqn:DR:wtVIXsquared} are due to an application of put-call parity.
\end{remark}

Since the VIX is defined for a fixed time horizon, it is useful to introduce a Musiela-like parameterisation, which is the purpose of the following lemma.
One may argue that such a representation is unnecessary for $\tau\in(0,T^*-t]$, due to the fact that the VIX is always defined for a $30$-day time horizon. The representation is not entirely worthless, however, as it provides a starting point for the analysis of the dynamics of the VIX when the definition is based on a rolling portfolio (see Equation~\eqref{eqn:Defn:PortfolioVIX}).
%

\begin{lem}
    For all $\tau>0$, let
    \begin{equation}
        {Y}_t(\tau):=\hat\VIX^2_t(t+\tau),\,\,\forall t>0.\label{eqn:DR:Ydef}
    \end{equation}
    Then
    \begin{align*}
        dY_t(\tau)={}&\fr{1}{\tau}\left(\int_{0}^\infty \fr{1}{k^2}\left(\fr{\pr\Theta}{\pr t}(k,\tau,t,\vX_t)+\fr{\pr\Theta}{\pr \tau}(k,\tau,t,\vX_t)\right)\,m(dk)-\,Y_t(\tau)\right)\,dt\\
                {}&-\fr{1}{2\tau\,F_t^2}d\qv{F}{F}_t+\fr{1}{\tau}\sum_{i=0}^d\int_0^\infty\fr{1}{k^2}\fr{\pr \Theta}{\pr x_i}(k,\tau,t,\vX_t)\,m(dk)\,dX^i_t\nonumber\\
                {}&+\fr{1}{2\tau}\sum_{i,j=0}^d\int_0^\infty\fr{1}{k^2}\fr{\pr^2 \Theta}{\pr x_i\pr x_j}(k,\tau,t,\vX_t)\,m(dk)\,d\qv{X^j}{X^i}_t,
    \end{align*}
    $\mP$-a.s..
\begin{proof}
    Equation~\eqref{eqn:DR:wtVIXsquared} and Equation~\eqref{eqn:DR:Ydef} imply the dynamics
    \begin{align*}
    {Y}_t(\tau)={}&Y_0(\tau)+\int_0^t\,dY_s(\tau)\\
                ={}&Y_0(\tau)+\int_0^t  d\hat\VIX^2_s(s+\tau)+\int_0^t\fr{\pr\hat\VIX^2_s}{\pr T}(s+\tau)\,ds\\
                ={}&Y_0(\tau)-\int_0^t\fr{1}{\tau}\hat\VIX^2_s(\tau)\,ds
                +\int_0^t\fr{1}{\tau}\int_{0}^\infty\fr{1}{k^2}\fr{\pr\hat\Theta}{\pr t}(k,s+\tau,s,\vX_s)\,m(dk)\,ds
                \\{}&+\sum_{i=0}^d\int_0^t\fr{1}{\tau}\int_0^\infty\fr{1}{k^2}\fr{\pr\hat \Theta}{\pr x_i}(k,s+\tau,s,\vX_s)\,m(dk)\,dX^i_s\\
                {}&+\fr{1}{2}\sum_{i,j=0}^d\int_0^t\fr{1}{\tau}\int_0^\infty\fr{1}{k^2}\fr{\pr^2\hat \Theta}{\pr x_i\pr x_j}(k,s+\tau,s,\vX_s)\,m(dk)\,d\qv{X^j}{X^i}_s\\
                {}&-\fr{1}{2}\int_0^t\fr{1}{\tau F_s^2}d\qv{F}{F}_s+\int_0^t\fr{\pr }{\pr T}\fr{1}{\tau}\int_0^\infty\fr{1}{k^2}{\hat\Theta}(k,s+\tau,s,\vX_s)\,m(dk)\,ds.
    \end{align*}
Writing the option prices as a function of time-till expiry and differentiating under the integral, which is a consequence of the Dominated Convergence Theorem, yields
    \begin{align*}
               {Y}_t(\tau) ={}&Y_0(\tau)-\int_0^t \fr{1}{\tau}\,Y_s(\tau)\,ds
                +\int_0^t\fr{1}{\tau}\int_{0}^\infty\fr{1}{k^2}\left(\fr{\pr\Theta}{\pr t}(k,\tau,s,\vX_s)+\fr{\pr\Theta}{\pr \tau}(k,\tau,s,\vX_s)\right)\,m(dk)\,ds\nonumber
                \\{}&-\fr{1}{2\tau}\int_0^t\fr{1}{F_s^2}d\qv{F}{F}_s+\sum_{i=0}^d\int_0^t\fr{1}{\tau}\int_0^\infty\fr{1}{k^2}\fr{\pr \Theta}{\pr x_i}(k,\tau,s,\vX_s)\,m(dk)\,dX^i_s\\
                {}&+\fr{1}{2}\sum_{i,j=0}^d\int_0^t\fr{1}{\tau}\int_0^\infty\fr{1}{k^2}\fr{\pr^2 \Theta}{\pr x_i\pr x_j}(k,\tau,s,\vX_s)\,m(dk)\,d\qv{X^j}{X^i}_s.
    \end{align*}

Expressing the above equation in the corresponding stochastic differential equation form completes the proof.
\end{proof}
\end{lem}

\begin{defn}
    Let $Y^*_t = Y_t(\tau^*)$.  Since $\tau^*$ is a fixed, the dependence of $Y^*_t$ on $\tau^*$ is suppressed.
\end{defn}

The object of concern is the VIX, not the square of the VIX, and the following proposition provides the dynamics of the VIX.
\begin{prop}
    \label{prop:VIXSDE:VIXSDE}
    Let $v(y):=\sqrt{y},\,\,\forall y>0$.  Therefore $V_t=v(Y^*_t),\,\,\forall t>0$, and the dynamics of $V_t$ are given by
    \begin{align}
        d\VIX_t &= u^{(1)}_t\VIX_t dt +u^{(2)}_t\VIX_t\,d\qv{F}{F}_t +\sum_{i,j=0}^du^{(i,j)}_t\VIX_t\,d\qv{X^j}{X^i}_t +\sum_{i=0}^dw^{(i)}_t\VIX_t\,dX^i_t,
    \end{align}
    where the drift and diffusion coefficients of $\VIX$ are given by the equations
    \begin{align}
        u^{(1)}_t ={}& -\fr{1}{2\tau^*}+\fr{1}{2\tau^*\VIX^2_t }\int_{0}^\infty \fr{1}{k^2}\left(\fr{\pr\Theta}{\pr t}(k,\tau^*,t,\vX_t)+\fr{\pr\Theta}{\pr \tau}(k,\tau^*,t,\vX_t)\right)\,m(dk),\label{eqn:VIM:u}\\
        u^{(2)}_t={}&-\fr{1}{4\tau^*\VIX^2_tF_t^2}\\
        u^{(i,j)}_t ={}&\fr{1}{4\tau^*\VIX_t^2}\int_{0}^\infty \fr{1}{k^2}\fr{\pr^2\Theta}{\pr x_i\pr x_j}(k,\tau^*,t,\vX_t)\,m(dk)-\fr{w^{(i)}_tw^{(j)}_t}{2}, \label{eqn:VIM:u2}\\{}\nonumber\\
         w^{(i)}_t ={}& \fr{1}{2\tau^*\VIX^2_t }\int_{0}^\infty \fr{1}{k^2}\fr{\pr\Theta}{\pr x_i}(k,\tau^*,t,\vX_t)\,m(dk),\label{eqn:VIM:w}
    \end{align}
    $\mP$-a.s., for all $t>0$.
\begin{proof}
    By Ito's formula,
    \begin{equation*}
        d\VIX_t = \fr{\pr v}{\pr y}(Y^*_t)\,dY^*_t+\fr{1}{2}\fr{\pr^2 v}{\pr y^2}(Y^*_t)\,d\qv{Y^*}{Y^*}_t,
    \end{equation*}
    where
    \begin{equation*}
         \fr{\pr v }{\pr y} (Y^*_t)= \fr{1}{2\VIX_t}\hspace{0.2cm}\mbox{ and } \hspace{0.2cm} \fr{\pr^2v}{\pr y^2} (Y^*_t)= -\fr{1}{4\VIX^3_t}.
    \end{equation*}
    Therefore
    \begin{align*}
        d\VIX_t ={}&\fr{1}{2\tau^*}\left(\fr{1}{\VIX_t }\int_{0}^\infty \fr{1}{k^2}\left(\fr{\pr\Theta}{\pr t}(k,\tau^*,t,\vX_t)+\fr{\pr\Theta}{\pr \tau}(k,\tau^*,t,\vX_t)\right)\,m(dk)-\VIX_t\right)\,dt\\
        {}&-\fr{1}{4\tau^*\VIX_tF_t^2}d\qv{F}{F}_t+\fr{1}{4\tau^*\VIX_t}\sum_{i,j=0}^d\left(\int_{0}^\infty \fr{1}{k^2}\fr{\pr^2\Theta}{\pr x_i\pr x_j}(k,\tau^*,t,\vX_t)\,m(dk)\right)\,d\qv{X^j}{X^i}_t\\
                {}&-\fr{1}{8(\tau^*)^2\VIX^3_t }\sum_{i,j=0}^d\left(\int_{0}^\infty \fr{1}{k^2}\fr{\pr\Theta}{\pr x_i}(k,\tau^*,t,\vX_t)\,m(dk)\right)\left(\int_{0}^\infty \fr{1}{k^2}\fr{\pr\Theta}{\pr x_j}(k,\tau^*,t,\vX_t)\,m(dk)\right)\,d\qv{X^j}{X^i}_t\\
                {}&+\fr{1}{2\tau^*\VIX_t }\sum_{i=0}^d\int_{0}^\infty \fr{1}{k^2}\fr{\pr\Theta}{\pr x_i}(k,\tau^*,t,\vX_t)\,m(dk)\,dX^i_t,\quad\mP\mbox{-a.s..}
    \end{align*}
    Introducing $u^{(1)}_t$, $u^{(2)}_t$, $u_t^{(i,j)}$ and $w^{(i)}_t$, as defined in Equations \eqref{eqn:VIM:u}-\eqref{eqn:VIM:w}, completes the proof.
\end{proof}
\end{prop}
%

\begin{cor}
   \label{cor:VFM:IndexToVIX}
    The dynamics of the index as specified in (A1) in Section~\ref{sect:VFM:Setup} imply that $\VIX_t$ satisfies
    \begin{align}
        \fr{d\VIX_t}{\VIX_t} &= u_t\, dt +\sum_{i=0}^d\,w^{(i)}_t\,X^i_t\,\sigma_t^i\cdot dW^\mP_t,\label{eqn:VFM:DerivedVIXSDE}
    \end{align}
    where the drift and diffusion coefficients of $\VIX$ are given by the equations
    \begin{align}
        u_t ={}& \sum_{i=0}^dw^{(i)}_t\,X_t^i\,\mu_t^i-\fr{1}{2\tau^*}+\fr{1}{2\tau^*\VIX^2_t }\int_{0}^\infty \fr{1}{k^2}\left(\fr{\pr\Theta}{\pr t}(k,\tau^*,t,\vX_t)+\fr{\pr\Theta}{\pr \tau}(k,\tau^*,t,\vX_t)\right)m(dk)-\fr{\sigma^0_t\cdot\sigma_t^0}{4\tau^*\VIX_t^2}\nonumber\\
        {}&+\sum_{i,j=0}^d\,X_t^i X_t^j\left(\fr{1}{4\tau^*\VIX_t^2}\int_{0}^\infty \fr{1}{k^2}\fr{\pr^2\Theta}{\pr x_i\pr x_j}(k,\tau^*,t,\vX_t)m(dk)-\fr{w^{(i)}_tw^{(j)}_t}{2}\right)\,\left(\sigma^i_t\cdot\sigma^j_t\right)
    \end{align}
    and
    \begin{align}
        w^{(i)}_t ={}& \fr{1}{2\tau^*\VIX^2_t }\int_{0}^\infty \fr{1}{k^2}\fr{\pr\Theta}{\pr x_i}(k,\tau^*,t,\vX_t)\,m(dk).\label{eqn:VFM:w}
    \end{align}
    \begin{proof}
        The result is an immediate consequence of Proposition~\ref{prop:VIXSDE:VIXSDE}, Equation~\eqref{eqn:VFM:Xvect} and Equation~\eqref{eqn:VFM:F_SDE}.
    \end{proof}
\end{cor}

\begin{remark}
    Given the form of the coefficients, it is not immediately obvious that there exist solutions to the Equation~\eqref{eqn:VFM:DerivedVIXSDE}.  Under very mild assumptions, however, $V$ is a strictly-positive process and will consequently possess a semi-martingale representation with local solutions. To ensure that $V$ has global solutions and a finite expectation, which are useful for practical purposes, additional assumptions are required.
\end{remark}

\begin{remark}
    Many modelling approaches begin by assuming a semi-martingale representation for the VIX. To obtain such a representation from first principles, assumptions must be made regarding the option price processes to ensure that fundamental quantities exist and are well defined, as illustrated by the analysis performed in this section.  Despite having an appearance of simplicity, the results of this section demonstrate that there are hidden complexities associated with modelling the VIX directly.  In modelling the VIX directly, many implicit assumptions regarding fundamental quantities are made.
\end{remark}

\subsection{An alternate representation of the diffusion term}

Proposition~\ref{prop:VIXSDE:VIXSDE} provides a representation of the VIX dynamics given very general dynamics for the index future and index options.
In what follows, an alternate representation of the diffusion coefficients,

\begin{align*}
    w^{(i)}_t ={}& \fr{1}{2\tau^*\VIX^2_t }\int_{0}^\infty \fr{1}{k^2}\fr{\pr\Theta}{\pr x_i}(k,\tau^*,t,\vX_t)\,m(dk),
\end{align*}
 given in Equation~\eqref{eqn:VIM:w}, is proposed.  The following analysis is not affected by an application of Girsanov's theorem, due to the fact that the multiplicative term is measure invariant, and is hence independent of any assumptions regarding market completeness and the market price of risk.

The following corollary offers an alternate representation of the diffusion term and is required for the forthcoming analysis in Section~\ref{sect:VFM:Consistency}.

\begin{cor}
\label{cor:VMM:diffusion}
Suppose that the measure $m(\cdot)$ is chosen as specified in Definition~\ref{def:Intro:CS}.  Further suppose that there exists a measure $\mQ$, such that index futures and options are $\mQ$-martingales.  Then the diffusion term for the VIX can be expressed as
\begin{align*}
  w^{(i)}_t ={}& -\fr{N}{2\tau^*\VIX^2_t }\fr{\pr}{\pr x_i}\left(\Ex^\mQ\left[\left.\ln(F_{t+\tau^*})\right|\vX_t=\vx\right]-\ln(F_t)\right).
\end{align*}
\begin{proof}
    The proof is almost identical to that of Lemma~\ref{cor:3on2:VIXasLog}.  Equation~\eqref{eqn:Intro:proofVIXasLog} in the proof of Lemma~\ref{cor:3on2:VIXasLog} implies that
\begin{align*}
  \int_E^\infty\fr{1}{k^2}\,C(k,\tau^*,t,\vX_t)\,dk+\int_0^E\fr{1}{k^2}\,P(k,\tau^*,t,\vX_t)\,dk = -\ExQ{\ln(F_{t+\tau^*})}+\ln(F_t)+\fr{1}{E}(F_{t}-E).
\end{align*}
The conditional expectation on the right-hand side can be expressed in terms of the value of the process $\vX$ at time $t$, as a consequence of the Markov property.  Choosing $E=F_t$ and differentiating with respect to $x_i$ yields
\begin{align*}
  \fr{\pr}{\pr x_i}\left.\int_0^\infty\fr{1}{k^2}\,\Theta(k,\tau^*,t,\vx)\,dk\right|_{\vx=\vX_t} = -\fr{\pr}{\pr x_i}\left(\Ex^\mQ\left[\ln(F_{t+\tau})|\vX_t=\vx\right]-\ln(F_t)\right).
\end{align*}
 Multiplying both sides by $N$, dividing by $2\tau^*V_t^2$, and differentiating under the integral, which valid as a consequence of the Dominated Convergence Theorem, completes the proof.
\end{proof}
\end{cor}

\begin{remark}
    The market practice is heavily dependent on the notion of Black implied volatility, which is a convention that has carried across from more traditional derivative products.  Traditional implied volatility is well studied and it has many well-documented complexities.  VIX implied volatilities come with many additional complexities, due to the connection between the underlying index and the VIX, and little progress in the way of mathematical results has been made.  The representation of the diffusion term in Corollary~\ref{cor:VMM:diffusion} provides a potential first step for such an analysis.
\end{remark}

\section{Deriving the Dynamics of the VIX from the VIX Futures}
\label{sect:VFM:ImpliedDynamics}
In this section, the object of concern is the family of equations presented in Equation~\eqref{eqn:VFM:F} for the VIX futures.  The implied dynamics of the VIX is the focus of this section and the analysis is independent of any specification for the underlying index.  To derive arbitrage restrictions on the joint market of VIX and equity derivatives, the dynamics of the VIX implied by the VIX futures must first be derived.  In the following proposition, the dynamics of the VIX are derived from the family of equations for the VIX futures.  The procedure is analogous to the recovery of the short-term interest rate from the forward-rate curve, which is a well known result in interest-rate modelling.

\begin{prop}
    \label{prop:VFM:VIXdiffusion}
    Suppose that the setup of (A3) in Section~\ref{sect:VFM:Setup} is assumed and that the coefficients $\mu^\VIX(t,T)$, $\nu(t,T)$ and the initial VIX futures term structure, $\vF(0,T)$, are differentiable with respect to $T$, with bounded partial derivatives $\mu^\VIX_T(t,T)$, $\nu_T(t,T)$ and $\vF_T(0,T)$. Then the dynamics of the implied VIX are given by
    \begin{equation}
    \iVIX_t = \iVIX_0+ \int_0^t \xi_u\iVIX_u\,du +\int_0^t\iVIX_u\nu(u,u)\cdot dW^{\mP}_u,\label{eqn:VFM:iVIX}
    \end{equation}
    where $\xi$ denotes the following process
    \begin{align}
        \xi_t ={}& \fr{\vF_T(0,t)}{\vF(t,t)}+\fr{1}{\vF(t,t)}\int_0^t\mu^\VIX(u,t)\vF(u,t)\,du\nonumber\\
        {}&+\fr{1}{\vF(t,t)}\int_0^t (\vF(u,t)\,\nu_T(u,t)+\vF_T(u,t)\,\nu(u,t))\cdot\,dW^\mP_u.\label{eqn:VFM:xi}
    \end{align}
\begin{proof}
    Recall that the implied VIX satisfies
    \begin{equation*}
    \iVIX_t=\vF(t,t)=\vF(0,t)+\int_0^t\mu^\VIX(u,t)\,\vF(u,t)\,du +\int_0^t\vF(u,t)\,\nu(u,t)\,\cdot dW^{\mP}_u.
    \end{equation*}
    Therefore
    \begin{align*}
        \iVIX_t ={}& \vF(0,t) +\int_0^t\mu^\VIX(u,t)\,\vF(u,t)\,du+\int_0^t\vF(u,u)\,\nu(u,u)\,\cdot dW^{\mP}_u\\
             {}&+\int_0^t\left(\vF(u,t)\,\nu(u,t) -\vF(u,u)\,\nu(u,u)\,\right)\,\cdot dW^{\mP}_u\\
             ={}& \vF(0,t)+\int_0^t\mu^\VIX(u,t)\,\vF(u,t)\,du +\int_0^t\iVIX_u\,\nu(u,u)\,\cdot dW^{\mP}_u\\
             {}&+\int_0^t\left(\vF(u,t)\,\nu(u,t) -\vF(u,u)\,\nu(u,u)\,\right)\,\cdot dW^{\mP}_u.
    \end{align*}
For fixed $u>0$ and for each $i=0,...,d$,
    \begin{align*}
        \vF(u,t)\,\nu(u,t){}& -\vF(u,u)\,\nu(u,u)\\
                ={}&\int_u^t \fr{d}{d s}[\vF(u,s)\nu(u,s)]ds\\
                 ={}& \int_u^t \left[\vF(u,s)\,\nu_T(u,s)+\vF_T(u,s)\nu(u,s)\right]ds.
    \end{align*}
Therefore
    \begin{align*}
             \iVIX_t={}& \vF(0,t)+\int_0^t\mu^\VIX(u,t)\,\vF(u,t)\,du +\int_0^t\iVIX_u\,\nu(u,u)\cdot dW^{\mP}_u\\
             {}&+\int_0^t \int_u^t \left[\vF(u,s)\,\nu_T(u,s)+\vF_T(u,s)\nu(u,s)\right]ds\cdot dW^{\mP}_u\\
             ={}&  \vF(0,t) +\int_0^t\mu^\VIX(u,t)\,\vF(u,t)\,du+\int_0^t\iVIX_u\,\nu(u,u)\cdot dW^{\mP}_u\\
             {}&+\int_0^t \int_0^s \left[\vF(u,s)\,\nu_T(u,s)+\vF_T(u,s)\nu(u,s)\right]\cdot dW^{\mP}_u \,ds,
    \end{align*}
where stochastic Fubini's theorem has been used to interchange the order of integration.
Writing
    \begin{align*}
         \vF(0,t) = \iVIX_0 + \int_0^t\vF_T(0,u)\,du
    \end{align*}
and introducing the process $\xi$ as defined in Equation~\eqref{eqn:VFM:xi} completes the proof.
\end{proof}
\end{prop}
%

%

\section{Consistency Conditions}
\label{sect:VFM:Consistency}
%
In this section, restrictions on the dynamics of the underlying index and the family of processes for the VIX futures are derived.  These restrictions are necessary for there to be no arbitrage between the joint market of VIX and equity derivatives and are referred to as \textit{consistency conditions}.   The conditions are formulated precisely in the following two definitions.

\begin{condition}[\textbf{C1}]\label{defn:VFM:C1}
    $\mP\,(\VIX_t=\iVIX_t)=1$, for all $0\le t\le T^*$, $\mP-a.s.$, where $\iVIX$ is the process given in Equation~\eqref{eqn:VFM:iVIX}.
\end{condition}
\begin{condition}[\textbf{C2}]\label{defn:VFM:C2}
    There exists an equivalent martingale measure, $\mQ\sim\mP$, for the underlying index and the VIX, such that futures and options on the index and futures on the VIX are $\mQ$-martingales.
\end{condition}

The first condition is a consequence of the restriction that the VIX and the dynamics implied by the VIX futures must be versions of the same process.  This condition simply requires the two different processes that one could derive for the VIX to be in agreement.  The second condition is a standard no arbitrage condition and it is obtained through an application of Girsanov's theorem.  The consistency conditions are equivalent to the forthcoming Theorem~\ref{thm:VFM:EMM1} and Theorem~\ref{thm:VFM:EMM2}.  The starting point for Theorem~\ref{thm:VFM:EMM1} is an equivalent martingale measure, from which the drift and diffusion restrictions are derived.  Theorem~\ref{thm:VFM:EMM2} starts with the drift and diffusion restrictions and is concerned with the existence of an equivalent martingale measure.  
%
%
\begin{thm}
\label{thm:VFM:EMM1}
    Suppose that there exists a measure $\mQ\sim\mP$ such that futures on the index and futures on the VIX are $\mQ$-martingales.  Then there exists a market price of risk, $\lambda$, with $\int_0^{T^*}|\lambda_s|^2\,ds<\infty$, $\mP$-a.s., such that for all $T\in[0,T^*-\tau^*]$ and $t\in[0,T]$,
\begin{align}
    \mu^0_t ={}& -\lambda_t\cdot\sigma^0_t,\label{eqn:VFM:CC1}
\end{align}
\begin{equation}
        \mu^\VIX(t,T) = -\lambda_t\cdot\nu(t,T),\label{eqn:VFM:CC2}
\end{equation}
for each $j=0,...,d$,
\begin{align}
      \nu^j(t,t)={}& \sum_{i=0}^d\left(\fr{X_t^i}{2\tau^*\VIX^2_t }\int_{0}^\infty \fr{1}{k^2}\fr{\pr\Theta}{\pr x_i}(k,\tau^*,t,\vX_t)\,m(dk)\,\sigma_t^{i,j}\right),\label{eqn:VFM:CC3}
\end{align}
with $\sigma^i = (\sigma^{i,1},\sigma^{i,2},...,\sigma^{i,d})^T$ and
\begin{align}
        \fr{1}{2\tau^*}+\fr{\sigma^0_t\cdot\sigma_t^0}{4\tau^*\VIX_t^2}+\sum_{i,j=0}\fr{1}{2}\,&{X_t^iX_t^j\,w^{(i)}_tw^{(j)}_t}\,\left(\sigma^i_t\cdot\sigma^j_t\right)
                +\lambda_t\cdot\sum_{i=0}^d\,w_t^i\,X_t^i\,\sigma^i_t\nonumber\\
                ={}&\fr{1}{\vF(t,t)}\int_0^t\left(\lambda_t\cdot\nu(t,t)\right)\vF(u,t)\,du-\fr{\vF_T(0,t)}{\vF(t,t)}\nonumber\\
        {}&-\fr{1}{\vF(t,t)}\int_0^t (\vF(u,t)\,\nu_T(u,t)+\vF_T(u,t)\,\nu(u,t))\cdot\,dW^\mP_u,\quad\mP\mbox{-a.s..}\label{eqn:VFM:CC4}
\end{align}

\end{thm}
\begin{proof}
\noindent Let $\mQ\sim\mP$ be an equivalent martingale measure on $(\Omega,\filt{T^*})$.  Then the existence of a $\lambda$, with $\int_0^{T^*}|\lambda_s|^2\,ds<\infty$, $\mP$-a.s., is a direct consequence of the Radon-Nikod\'ym theorem.
By Girsanov's theorem
    \begin{equation}
        W^\mQ_t = W^\mP_t - \int_0^t\lambda_s\,ds\label{eqn:VFM:measurechange}
    \end{equation}
    defines a Brownian motion under $\mQ$.
    Therefore
    \begin{align*}
          \fr{dF_t}{F_t}  &= \mu^0_tdt +\sigma^0_t\cdot\left(dW^\mQ_t+\lambda_tdt\right)\\
                &= \left(\mu^0_t+\lambda_t\cdot\sigma^0_t\right)dt +\sigma^0_t\cdot dW^\mQ_t
    \end{align*}
    and
    \begin{align*}
    d\hat\Theta(k,{}&T,t,\vX_t) \\={}& \fr{\pr\hat\Theta}{\pr t}(k,T,t,\vX_t)\,dt+\sum_{i=0}^d\fr{\pr\hat \Theta}{\pr x_i}(k,T,t,\vX_t)\,dX^i_t+\fr{1}{2}\sum_{i,j=0}^d\fr{\pr^2\hat \Theta}{\pr x_i\pr x_j}(k,T,t,\vX_t)\,d\qv{X^j}{X^i}_t \\
          ={}& \fr{\pr\hat\Theta}{\pr t}(k,T,t,\vX_t)\,dt
            +\sum_{i=0}^d\fr{\pr\hat \Theta}{\pr x_i}(k,T,t,\vX_t)X_t^i\,\mu^i_t\,dt
            \\{}&+\fr{1}{2}\sum_{i,j=0}^d\fr{\pr^2\hat \Theta}{\pr x_i\pr x_j}(k,T,t,\vX_t)\,d\qv{X^j}{X^i}_t+\sum_{i=0}^d\fr{\pr\hat\Theta}{\pr x_i}(k,T,t,\vX_t)X_t^i\sigma^i_t\cdot dW^\mP_t\\
          ={}& \fr{\pr\hat\Theta}{\pr t}(k,T,t,\vX_t)\,dt
            +\sum_{i=0}^d\fr{\pr\hat \Theta}{\pr x_i}(k,T,t,\vX_t)\left(\mu^i_t+\lambda_t\cdot\sigma_t\right)X_t^i\,dt
            \\{}&
            +\fr{1}{2}\sum_{i,j=0}^d\fr{\pr^2\hat \Theta}{\pr x_i\pr x_j}(k,T,t,\vX_t)\,d\qv{X^j}{X^i}_t+\sum_{i=0}^d\fr{\pr\hat\Theta}{\pr x_i}(k,T,t,\vX_t)X_t^i\sigma^i_t\cdot dW^\mQ_t.
    \end{align*}
    The no-arbitrage condition, stated as Condition~\ref{defn:VFM:C2} \textbf{(C2)}, implies that, for all $t\in[0,T^*]$,
    \begin{equation*}
        \mu^0_t = -\lambda_t\cdot\sigma^0_t,\quad \mQ-a.s.,
    \end{equation*}
    and that for all $k>0$ and each $T\in(t,T^*]$,
    \begin{align*}
        \int_0^t\left(\fr{\pr\hat\Theta}{\pr t}\right.&(k,T,s,\vX_s)\,ds
            +\sum_{i=0}^d\fr{\pr\hat \Theta}{\pr x_i}(k,T,s,\vX_s)X_s^i\,\mu^i_s\,ds\\
            {}&\hspace{-0.7cm}+\left.\fr{1}{2}\sum_{i,j=0}^d\fr{\pr^2\hat \Theta}{\pr x_i\pr x_j}(k,T,s,\vX_s)\,d\qv{X^j}{X^i}_s\right)
            =-\int_0^t\left(\lambda_s\cdot\sum_{i=0}^d\fr{\pr\hat \Theta}{\pr x_i}(k,T,s,\vX_s)X_s^i\sigma^i_s\right)\,ds,
    \end{align*}
    $\mQ-a.s.$.  Introducing the change of variable $T:=t+\tau$ and expressing option prices as a function of time-till expiry, $\Theta(\cdot,\tau,\cdot,\cdot):=\hat\Theta(\cdot,t+\tau,\cdot,\cdot)$, for all $\tau\in[0,T^*]$ and $t\in[0,T^*-\tau]$, gives
    \begin{align}
        \int_0^t\left(\fr{\pr\Theta}{\pr t}\right.&(k,\tau,s,\vX_s)\,ds
            +\fr{\pr\Theta}{\pr \tau}(k,\tau,s,\vX_s)\,ds\nonumber
            +\sum_{i=0}^d\fr{\pr \Theta}{\pr x_i}(k,\tau,s,\vX_s)X_s^i\,\mu^i_s\,ds\\
            {}&\hspace{-0.7cm}+\left.\fr{1}{2}\sum_{i,j=0}^d\fr{\pr^2 \Theta}{\pr x_i\pr x_j}(k,\tau,s,\vX_s)\,d\qv{X^j}{X^i}_s\right)
            =-\int_0^t\left(\lambda_s\cdot\sum_{i=0}^d\fr{\pr \Theta}{\pr x_i}(k,\tau,s,\vX_s)X_s^i\sigma^i_s\right)\,ds,\label{eqn:VFM:ThetaDrift}
    \end{align}
    for all $k>0$.  Corollary~\ref{cor:VFM:IndexToVIX} and Equation~\eqref{eqn:VFM:ThetaDrift} evaluated at $\tau=\tau^*$ imply that
    \begin{align}
\fr{d\VIX_t}{\VIX_t} ={}&-\left[\fr{1}{2\tau^*}+\fr{\sigma^0_t\cdot\sigma_t^0}{4\tau^*\VIX_t^2}+\sum_{i,j=0}\fr{1}{2}\,{X_t^iX_t^j\,w^{(i)}_tw^{(j)}_t}\,\left(\sigma^i_t\cdot\sigma^j_t\right)
                +\lambda_t\cdot\sum_{i=0}^d\,w_t^i\,X_t^i\sigma^i_t\right]\,dt\nonumber\\
                {}&+\sum_{i=0}^dw^{(i)}_t\,X_t^i\,\sigma_t^i\cdot dW^\mP_t.\label{eqn:VFM:VIXSDE1}
    \end{align}

    The VIX futures dynamics given by the family of processes in Equation~\eqref{eqn:VFM:F} satisfy, under the risk-neutral measure $\mQ$,
    \begin{align*}
       d\vF(t,T) ={}& \left(\mu^\VIX(t,T)+\lambda_t\cdot\nu(t,T)\right)\vF(t,T)\,dt+ \vF(t,T)\,\nu(t,T)\cdot dW^\mQ_t.
    \end{align*}
    The no-arbitrage condition, stated as Condition~\ref{defn:VFM:C2} \textbf{(C2)}, again implies that, for all $t\in[0,T]$,
    \begin{equation}
        \mu^\VIX(t,T) = -\lambda_t\cdot\nu(t,T),\quad \mQ-a.s..\label{eqn:VFM:VfutDR}
    \end{equation}
    Proposition~\ref{prop:VFM:VIXdiffusion} and Equation~\eqref{eqn:VFM:VfutDR} imply that the VIX satisfies the dynamics
    \begin{equation}
        d\iVIX_t = \xi_t\iVIX_t\,dt +\iVIX_t\,\nu(t,t)\cdot dW^\mQ_t,
        \label{eqn:VFM:VIXSDE2}
    \end{equation}
    where the process $\xi$ is defined by
    \begin{align*}
        \xi_t ={}& \fr{\vF_T(0,t)}{\vF(t,t)}-\fr{1}{\vF(t,t)}\int_0^t\left(\lambda_t\cdot\nu(t,t)\right)\vF(u,t)\,du\nonumber\\
        {}&+\fr{1}{\vF(t,t)}\int_0^t (\vF(u,t)\,\nu_T(u,t)+\vF_T(u,t)\,\nu(u,t))\cdot\,dW^\mP_u.
    \end{align*}
    Condition~\ref{defn:VFM:C1} \textbf{(C1)} requires for the VIX and the dynamics of the implied VIX to be versions of the same process.  Imposing the condition that the drift and diffusion coefficients in Equation~\eqref{eqn:VFM:VIXSDE1} and Equation~\eqref{eqn:VFM:VIXSDE2} must be equal completes the proof.
\end{proof}
\begin{thm}
        \label{thm:VFM:EMM2}
        Suppose that $\mu$, $\sigma$, $\xi$ and $\nu$ satisfy, as functions of $F$ and $\vF$, Equations~\eqref{eqn:VFM:CC1}-\eqref{eqn:VFM:CC4}, for all $T\in [0,T^*-\tau^*]$ and all $t\in[0,T]$, $\mP$-a.s., for some process $\lambda$, with $\int_0^{T^*}|\lambda_s|^2\,ds<\infty$, $\mP$-a.s., and
       \[\Ex^{\mP}\left[\left.\mathcal{E}\left(\int_0^{\cdot} \lambda_s dW^\mP_s\right)_{T^*}\right|\filt{0}\right]=1.\]
       Further suppose that there exists an adapted process $F$ on $[0,T^*-\tau^*]$ and a family of adapted processes $\vF(\cdot,T)$, for all $T\in[0,T^*-\tau^*]$, on $[0,T]$ satisfying Equation~\eqref{eqn:VFM:F_SDE} and Equation~\eqref{eqn:VFM:F}. Then there exists an equivalent measure, $\mQ\sim\mP$, on $\filt{T^*}$, for $(F_t)_{0\le t\le T^*}$ and $(\vF(t,T))_{0\le t\le T}$, for all $T\in[0,T^*]$, such that futures on the index and futures on the VIX are $\mQ$-martingales.
\end{thm}
\begin{remark}
    One such measure is given by
\begin{equation}
  \fr{d \mQ}{d \mP} := \mathcal{E}\left(\int_0^{\cdot} \lambda_s dW^\mP_s\right)_{T^*},
\end{equation}
where $\mathcal{E}(\cdot)$ denotes the stochastic exponential.
\end{remark}
\begin{proof}
    The existence of a measure $\mQ\sim\mP$ is obtained through a direct application of Girsanov's Theorem.  The proof of the remainder of Theorem~\ref{thm:VFM:EMM2} follows from the proof of Theorem~\ref{thm:VFM:EMM1}.
\end{proof}
\begin{remark}
The martingale measure of Theorem \ref{thm:VFM:EMM2} may not be unique due to the fact that there may be sources of risk that are not traded. From a practical perspective, the ability to replicate contingent claims is potentially of more interest than the theoretical completeness of a model.  Theorem \ref{thm:VFM:EMM1} and Theorem~\ref{thm:VFM:EMM2} provide an arbitrage-free representation that enables the direct hedging of VIX options with VIX futures contracts.
\end{remark}
\begin{remark}
    Restrictions on both the drift and diffusion coefficients of the index and VIX futures are imposed in Theorem \ref{thm:VFM:EMM1} and Theorem~\ref{thm:VFM:EMM2}.  The non-uniqueness of the martingale measure enters the theorem through the presence of $\lambda$ in the drift restrictions. The diffusion coefficient, however, is not affected by an application of Girsanov's theorem, due to the fact that the multiplicative term is measure invariant.  An analysis of the diffusion coefficients is hence independent of any assumptions regarding market completeness and the market price of risk.
\end{remark}
The restrictions on the diffusion terms are rather non-standard and the interpretation is not immediately clear, due to the presence of the term $V_t$.  The forthcoming Theorem~\ref{thm:VFM:genPDE} is concerned with the implications of Theorem~\ref{thm:VFM:EMM1} and Theorem~\ref{thm:VFM:EMM2} on the diffusion term for the VIX.  Before introducing the theorem, consider the following assumption regarding the diffusion terms.  The functional form is motivated from the perspective of specifying a model for the VIX futures.  The reason that the diffusion coefficient for the index is dependent upon $\VIX_t$, is that this may be an implicit consequence from the assumed dynamics for the VIX futures.
 \begin{assumption}
    The diffusion coefficient $\sigma$ is a function of the processes $X$ and $\VIX$, and the process $\nu$ is a functional of the VIX futures curve $\vF$, such that
    \begin{align*}
        \sigma_t=\sigma(t,\vX_t,\VIX_t),\quad\mbox{and}\quad\nu(t,T) = \nu(t,T,\vF(t,T)).
    \end{align*}
\end{assumption}

The partial-differential equation (PDE) in the following theorem provides a restriction on the choice of the diffusion coefficients, $\sigma$ and $\nu$, so that the joint dynamics are in agreement.  The result is somewhat of an inverse problem; the solution to the equation is given, while the coefficients are unknown.

\begin{thm}
\label{thm:VFM:genPDE}
    Suppose that the measure $m(\cdot)$ is chosen as specified in Definition~\ref{def:Intro:CS}.  Then, for there to be no joint arbitrage in the sense of Theorem~\ref{thm:VFM:EMM1} and Theorem~\ref{thm:VFM:EMM2}, for all $t\ge 0$ and $\vx\in\Rp^{d+1}$,
    \begin{align}
        \sum_{i=0}^d\,x^i\,\sigma^{i,j}(t,\vx,\sqrt{h(t,\vx)})\fr{\pr h}{\pr x_i}(t,\vx)= 2h(t,\vx)\,\nu^j(t,t,\sqrt {h(t,\vx)}),\quad j=0,...,d,\label{eqn:VFM:genPDE}
    \end{align}
    with $h(t,\vX_t)=V^2_t$.
    \begin{proof}
    Equation~\eqref{eqn:VFM:CC3} in Theorem \ref{thm:VFM:EMM1} states that for each $j=0,...,d$,
\begin{align}
      \nu^j(t,t,V_t)={}& \sum_{i=0}^d\left(\fr{X_t^i}{2\tau^*\VIX^2_t }\int_{0}^\infty \fr{1}{k^2}\fr{\pr\Theta}{\pr x_i}(k,\tau^*,t,\vX_t)\,m(dk)\,\sigma_t^{i,j}\right),\label{eqn:VFM:diffPDEproof}
\end{align}
    with $\sigma^i = (\sigma^{i,1},\sigma^{i,2},...,\sigma^{i,d})^T$. Recall Corollary~\ref{cor:VMM:diffusion}, which states the following alternate representation of the diffusion term
    \begin{align*}
  \fr{1}{2\tau^*\VIX^2_t }\int_{0}^\infty \fr{1}{k^2}\fr{\pr\Theta}{\pr x_i}(k,\tau^*,t,\vX_t)\,m(dk) ={}& -\fr{N}{2\tau^*\VIX^2_t }\fr{\pr}{\pr x_i}\left(\Ex^\mQ\left[\left.\ln(F_{t+\tau})\right|\vX_t=\vx\right]-\ln(F_t)\right).
\end{align*}
    Lemma~\ref{cor:3on2:VIXasLog} and the Markov property imply that the VIX can be represented as
    \begin{equation*}
        \VIX^2_t = {- \fr{N}{\tau^*}\Ex^\mQ\left[\ln  \left.\left(\frac{F_{t+\tau^*}}{F_t}\right)\right|\vX_t =\vx \right]}=:h(t,\vx),
    \end{equation*}
    for some function $h:[0,T^*]\times\Rp^{d+1}\to\Rp$, and hence
\begin{align}
      \nu^j(t,t,\sqrt{h(t,\vx)})={}& \sum_{i=0}^d\left(\fr{x^i}{2h(t,\vx) }\,\sigma^{i,j}(t,\vx,\sqrt{h(t,\vx)})\,\fr{\pr h}{\pr x_i}(t,\vx)\right),\quad j=0,...,d.
\end{align}
    Multiplying both sides by $2h(t,\vx)$ completes the proof.
\end{proof}

\end{thm}
\section{Application - Proportional volatility}
\label{sect:VFM:Applications}
In this section, an application of the main result is provided.  A class of model for the term structure of VIX futures is considered and Theorem~\ref{thm:VFM:genPDE} is used to derive the implications of the model assumptions on the dynamics of the underlying index, such that the joint market remains free from arbitrage.  The example demonstrates that there are unavoidable complexities involved when modelling the joint dynamics of the underlying index and VIX futures, and that care must be taken to avoid arbitrage.  %
The model is a special case of models that satisfy the restrictions derived in Section~\ref{sect:VFM:Consistency}, except that the risk-neutral measure $\mQ$ is assumed to be fixed and the dynamics are directly specified under such a measure.

The simplest non-negative class of model for the term structure of VIX futures is geometric Brownian motion.  The influence of randomness is to shift the entire futures term structure up or down in a multiplicative fashion and options are priced through a straight-forward application of Black's formula.  The example provides a good illustration of the implications of the modelling approach, as well as a first step in an analysis of VIX implied volatilities, since these are calculated using Black's formula.

In what follows, it is assumed that the underlying index is driven by a one-factor stochastic-volatility model.  More precisely, let $\mu:[0,\infty)\to\mathbb R$ and $\sigma:[0,\infty)\to\mathbb R$ be two functions. We assume that the dynamics of the index process $\left(F_t\right)$ is defined by a system of equations
 \begin{equation}\label{eqn:VFM:indexdiff2}
 \left\{\begin{aligned}
        \fr{dF_t}{F_t}&=\sqrt{X_t}\,dW_t\\
        dX_t &= \mu(X_t)\,dt+\sigma(X_t)\,dZ_t\,,
    \end{aligned}\right.
    \end{equation}
where $W$ and $Z$ are correlated $\mathbb{F}$-adapted Brownian motions under the measure $\mQ$.
We assume that  for every $X_0=x\ge 0$ and $F_0=f\ge 0$ there exists a unique strong solution $\left(F_t,X_t\right)$ to the system \eqref{eqn:VFM:indexdiff2} such that  $\mathbb P\left(X_t>0\right)=1$, for all $t\ge0$,
$\mQ$-a.s. In particular,
\[\int_0^t\left|\mu\left(X_s\right)\right|\,ds+\int_0^t\left|\sigma\left(X_s\right)\right|^2\,ds<\infty\quad\mQ-a.s.\]
\begin{cor}
%
\label{cor:VFM:propVol}
Suppose that VIX futures satisfy the family of equations
    \begin{equation}
        {d\vF(t,T)}=\beta(t,T){\vF(t,T)}\,dZ_t,\quad\forall 0\le t< T\le T^*, \label{eqn:VFM:propVol}
    \end{equation}
for some $\beta: U\to \Rp$ where $U=\{(t,T)\,|\,0\le t<T\le T^*\}$, and that the dynamics of the underlying index are given by the stochastic volatility model \eqref{eqn:VFM:indexdiff2}.

    Then, for there to be no arbitrage, $\beta(t,t)$, $\mu(x)$ and $\sigma(x)$ must satisfy
    \begin{equation*}
        \beta(t,t) \equiv \gamma
    \end{equation*}
and
    \begin{align*}
        \mu(x) &= \sigma(x)\left[\fr{1}{2}\fr{\pr \sigma}{\pr x}(x)-\gamma\right],
    \end{align*}
    for some $\gamma\in\Rp$ and for all $t\ge 0$ and $x\ge0$.

 \begin{proof}
%

    By Corollary~\ref{cor:3on2:VIXasLog}, the VIX satisfies
    \begin{equation*}
        \VIX_t = \sqrt{-\fr{N}{\tau^*}\ExQ{\ln \left( \frac{F_{t+\tau^*}}{F_t} \right)}}.
    \end{equation*}
     The dynamics specified in Equation~\eqref{eqn:VFM:indexdiff2} further imply that%
    \begin{align*}
        \VIX_t &= \sqrt{-\fr{N}{\tau^*}\ExQ{\int_t^{t+\tau^*}\fr{1}{F_{s}}\,dF_s-\int_{t}^{t+\tau^*}\fr{1}{2F_{s}^2}\,d\qv{F}{F}_s}}\\
             &= \sqrt{\fr{N}{2\tau^*}\ExQ{\int_{t}^{t+\tau^*}X_s\,ds}}\\
             &= \sqrt{\fr{N}{2\tau^*}\Ex^\mQ\left[\int_t^{t+\tau^*}X_s\,ds\bigg|X_t=x\right]}=: \sqrt{h(x)},
    \end{align*}
    for some function $h:\Rp\to\Rp$.
%
    The function $h(x)$ is independent of $t$ due to the Markovian structure of $X_t$, which can be observed by introducing a simple change of variable,
    \begin{align*}
            h(x) &=  \fr{N}{2\tau^*}\Ex^\mQ\left[\int_t^{t+\tau^*}X_s\,ds\bigg|X_t=x\right] \\
                & = \fr{N}{2\tau^*}\Ex^\mQ\left[\int_t^{t+\tau^*}X_{s-t}\,ds\bigg|X_0=x\right] \\
                & = \fr{N}{2\tau^*}\Ex^\mQ\left[\int_0^{\tau^*}X_{r}\,dr\bigg|X_0=x\right].
    \end{align*}
    To proceed, introduce the function
    \begin{equation}
        H(t,x):=  {\fr{N}{2\tau^*}}\Ex^\mQ\left[\int_t^{\tau^*}X_{r}\,dr\bigg|X_0=x\right],\,0\le t\le \tau^*,\label{eqn:VFM:H}
    \end{equation}
    such that $H(0,x)\equiv h(x)$ and the function $H(t,x)$ is the unique solution to the Cauchy problem (see Theorem~7.6 in \cite{KaratzasShreve00})
     \begin{equation*}
        \left\{
        \begin{array}{l}
        \ds{\fr{\pr H}{\pr t}(t,x)+\ds{\mu(x)\fr{\pr H}{\pr x}(t,x)+\fr{1}{2}\sigma^2(x)\fr{\pr ^2H}{\pr x^2}(t,x)+x= 0}},\quad 0\le t <\tau^*,\\
        H(\tau^*,x) = 0.
        \end{array}
        \right.
    \end{equation*}
    Differentiating Equation~\eqref{eqn:VFM:H} and taking the limit $t\searrow 0$ implies that $\fr{\pr H}{\pr t}(0^+,x)=-x$.  Evaluating the above equation at $t=0$ implies that $h(x)$ satisfies
    \begin{equation}
        \ds{\mu(x)\fr{\pr h}{\pr x}(x)+\fr{1}{2}\sigma^2(x)\fr{\pr ^2h}{\pr x^2}(x)= 0}.\label{eqn:VFM:app:FKv}
    \end{equation}
%
Theorem~\ref{thm:VFM:genPDE} implies that $\sigma(x)$ and $h(x)$ must jointly satisfy
    \begin{align}
        \sigma(x)\fr{\pr h}{\pr x}(x)= 2\beta(t,t)\,h(x),\label{eqn:VFM:propPDE}
    \end{align}
     for all $t\ge 0$ and $x\ge0$.  Using the fact that $h(x)$ is constant in time implies that $\beta(t,t)\equiv\gamma$, for some $\gamma>0$.
Differentiating Equation~\eqref{eqn:VFM:propPDE} with respect to $x$ yields
    \begin{align*}
        \fr{\pr\sigma}{\pr x}(x)\fr{\pr h}{\pr x}(x)+ \sigma(x)\fr{\pr^2 h}{\pr x^2}(x)&= 2\gamma\fr{\pr h}{\pr x}(x)
    \end{align*}
    and hence
    \begin{align}
      \fr{1}{2}\sigma^2(x)\fr{\pr^2 h}{\pr x^2}(x)&=\fr{1}{2}\,{\sigma(x)}\,\fr{\pr h}{\pr x}(x)\left[2\gamma -\fr{\pr\sigma}{\pr x}(x)\right]\nonumber\\
        &=\gamma\,h(x)\left[2\gamma-\fr{\pr\sigma}{\pr x}(x)\right].
        \label{eqn:VFM:App:propd2Udx2}
    \end{align}
Equations \eqref{eqn:VFM:app:FKv}-\eqref{eqn:VFM:App:propd2Udx2} imply that
\begin{equation}
    \ds{2\gamma\,h(x)\left[ \fr{\mu(x)}{\sigma(x)}+\gamma-\fr{1}{2}\fr{\pr\sigma}{\pr x}(x)\right]= 0}\label{eqn:VFM:app:FKuAfterSub}
\end{equation}
and hence
\begin{equation*}
     \fr{\mu(x)}{\sigma(x)}+\gamma-\fr{1}{2}\fr{\pr\sigma}{\pr x}(x)= 0.
\end{equation*}
Solving for $\mu(x)$ completes the proof.
    \end{proof}
\end{cor}
In the following corollary and example, it is shown that the assumptions of Corollary \ref{cor:VFM:propVol} are satisfied for a nontrivial class of volatilities $\sigma$.
\begin{cor}
Assume that $\sigma(x)=x\sigma_0(x)$, where $\sigma_0:\mathbb R\to\mathbb R$ is continuously differentiable and
\[\sup_{x\in\mathbb R}\left|\sigma_0(x)\right|+\left|\sigma^\prime(x)\right|<\infty,\quad x\in\mathbb R\,.\]
 Then, for every $X_0>0$, there exists a unique global non-negative solution of the equation
 \begin{equation}\label{eq_X}
            dX_t = \sigma(X_t)\left(\fr{1}{2}\sigma^\prime(X_t)-\gamma\right)dt+\sigma(X_t)\,dZ_t.
    \end{equation}
    Moreover, the process $F$ defined in Equation \eqref{eqn:VFM:indexdiff2} is positive martingale for every $F_0>0$.
    \end{cor}
       \begin{proof}
       The assumptions of the corollary yield the existence and uniqueness of solutions to the Equation~\eqref{eq_X} for every $x\in\mathbb R$. Moreover,
\[X_t=\exp\left(\int_0^t\left(\sigma_0\left(X_s\right)\left(\frac{1}{2}\sigma^\prime\left(X_s\right)-\gamma-\frac{1}{2}\sigma_0^2\left(X_s\right)\right)\right)ds +\int_0^t\sigma_0\left(X_s\right)dZ_s\right).\]
It is easy to check that
        \[\mathbb E\left(\int_0^TX_s\,ds\right)<\infty\,,\]
        hence $\left(F_t\right)$ is a square-integrable martingale.
            \end{proof}
            \par\medskip\noindent
            \emph{Example. }
            For a certain $\alpha>0$, consider Equation \eqref{eq_X} with $\sigma(x)=\alpha\sqrt{x}$, for $x\ge 0$. Then the equation takes the form
         \begin{equation*}
            dX_t = \fr{\alpha^2}{4}dt+\alpha\,\sqrt{X_t}\left(-\gamma dt+dZ_t\right).
    \end{equation*}
Recall that the equation
\[dY_t=\frac{\alpha^2}{4}dt+\alpha\sqrt{Y_t}dZ_t,\quad Y_0>0,\]
has unique strong solution that is instantaneously reflected at zero, hence $\mathbb P\left(Y_t\ge 0\right)=1$ for every $t>0$.
Using equivalent change of measure, Equation \eqref{eq_X} has a unique weak solution $X$, such that $\mathbb P\left(X_t\ge 0\right)=1$ for every $t>0$. Since $\gamma>0$ and
\[\mathbb E(X_t)+\alpha\gamma\int_0^t\mathbb E(\sqrt{X_s})ds=\mathbb E(X_0)+\frac{\alpha^2}{4}t\,\]
the martingale property of $\left(F_t\right)$ follows.


In Corollary~\ref{cor:VFM:propVol}, the plausible dynamics for the instantaneous variance process are derived from the associated dynamics for the term-structure of VIX futures. A similar analysis was performed in \cite{CarrSu07} in the context of variance swaps.  The authors propose a rather general framework in which the underlying index and a single variance swap are modelled.  Plausible risk-neutral dynamics for the instantaneous variance process are derived so that the dynamics of the underlying are consistent with that of the variance swap. Their restrictions are different to those obtained here, however, due to the fact that the object of concern is a variance swap, not the VIX.  The authors examine the limiting case as the variance swap approaches expiry and, together with PDE arguments, obtain interesting restrictions from a rather general setup.  The analogous arguments presented in Corollary~\ref{cor:VFM:propVol} for the VIX also make use of PDE arguments, however, where limiting arguments were used for variance swaps, time-invariant arguments were used, since the VIX is defined over a fixed time horizon.
\section{Conclusion}
In this paper, a new modelling approach that prescribes dynamics to the term structure of VIX futures was proposed.  The main contributions were Theorem~\ref{thm:VFM:EMM1} and Theorem~\ref{thm:VFM:EMM2}, which stated necessary conditions for there to be no arbitrage between the joint market of VIX and equity derivatives. Not surprisingly, the restrictions are rather complex, as a consequence of the complexities involved in the definition of the VIX.  An application of the main result was provided, which demonstrates that when modelling VIX futures directly, the drift and diffusion of the corresponding stochastic volatility model must be restricted to preclude arbitrage.  This is similar to the well-known drift restrictions in interest-rate modelling.

There are several directions in which the research could be extended. The analysis performed in this paper provides a platform for the analysis of the existing literature. For models that directly prescribe dynamics to the VIX, the affect of assuming a specific form for the VIX drift and diffusion coefficients can be assessed.  Moreover, the newly developed framework provides a starting point for the analysis of VIX option implied volatilities, which are a fundamental quantity for both academics and practitioners.

\appendix
\section{Proof of Proposition~\ref{prop:VFM:Vhat2}}
\label{app:proof:propVhat2}
\begin{proof}
For $0\le t<T\le T^*$ and fixed $f>0$, let
\begin{equation*}
    G(f,T,t,\vX_t) = \int_0^f\fr{1}{k^2}\hat P(k,T,t,\vX_t)\,m(dk)+\int_f^\infty\fr{1}{k^2}\hat C(k,T,t,\vX_t)\,m(dk)
\end{equation*}
and write $G_t(f):=G(f,T,t,\vX_t)$.  Then
\begin{equation}
    \hat\VIX^2_t(T) = \fr{1}{T-t}G_t(F_t).
\end{equation}
By the assumption that $G_t(\cdot)$ is twice differentiable,
\begin{align*}
    \fr{\pr G_t}{\pr f}(f) &=\fr{\hat P(f,T,t,\vX_t)}{f^2}-\fr{\hat C(f,T,t,\vX_t)}{f^2}
\end{align*}
and
\begin{align*}
 \fr{\pr^2G_t}{\pr f^2}(f) &=\fr{\pr}{\pr f}\left(\fr{\hat P(f,T,t,\vX_t)}{f^2}-\fr{\hat C(f,T,t,\vX_t)}{f^2}\right).
\end{align*}
Then, by Ito's Lemma and stochastic Fubini's theorem,
\begin{align*}
    G_t(f) ={}& G_0(f)+\int_0^f\fr{1}{k^2}\int_0^t d\hat P(k,T,s,\vX_s)\,m(dk)+\int_f^\infty\fr{1}{k^2}\int_0^td\hat C(k,T,s,\vX_s)\,m(dk)\\
                ={}&G_0(f)+ \int_0^f\int_0^t \fr{1}{k^2}\left(\fr{\pr\hat P}{\pr t}(k,T,s,\vX_s)\,ds+\sum_{i=0}^d\fr{\pr\hat P}{\pr x_i}(k,T,s,\vX_s)\,dX^i_s\right.\\
                {}&\left.+\fr{1}{2}\sum_{i,j=0}^d\fr{\pr^2\hat P}{\pr x_i\pr x_j}(k,T,s,\vX_s)\,d\qv{X^j}{X^i}_s\right)\,m(dk)
                +\int_f^\infty\int_0^t \fr{1}{k^2}\left(\fr{\pr\hat C}{\pr t}(k,T,s,\vX_s)\,ds\right.\\{}&\left.+\sum_{i=0}^d\fr{\pr\hat C}{\pr x_i}(k,T,s,\vX_s)\,dX^i_s+\fr{1}{2}\sum_{i,j=0}^d\fr{\pr^2\hat C}{\pr x_i\pr x_j}(k,T,s,\vX_s)\,d\qv{X^j}{X^i}_s\right)\,m(dk)\\
                ={}&G_0(f)+ \int_0^t\int_0^f \fr{1}{k^2}\fr{\pr\hat P}{\pr t}(k,T,s,\vX_s)\,m(dk)\,ds+\sum_{i=0}^d\int_0^t\int_0^f\fr{1}{k^2}\fr{\pr\hat P}{\pr x_i}(k,T,s,\vX_s)\,m(dk)\,dX^i_s\\
                {}&+\fr{1}{2}\sum_{i,j=0}^d\int_0^t\int_0^f\fr{1}{k^2}\fr{\pr^2\hat P}{\pr x_i\pr x_j}(k,T,s,\vX_s)\,m(dk)\,d\qv{X^j}{X^i}_s\\
                {}&+\int_0^t\int_f^\infty\fr{1}{k^2}\fr{\pr\hat C}{\pr t}(k,T,s,\vX_s)\,m(dk)\,ds+\sum_{i=0}^d\int_0^t\int_f^\infty\fr{1}{k^2}\fr{\pr\hat C}{\pr x_i}(k,T,s,\vX_s)\,m(dk)\,dX^i_s\\
                {}&+\fr{1}{2}\sum_{i,j=0}^d\int_0^t\int_f^\infty\fr{1}{k^2}\fr{\pr^2\hat C}{\pr x_i\pr x_j}(k,T,s,\vX_s)\,m(dk)\,d\qv{X^j}{X^i}_s
\end{align*}
and
\begin{equation}
    \hat\VIX^2_t(T)=\hat\VIX^2_0(T)-\int_0^t\fr{1}{T-s}\hat\VIX^2_s(T)\,ds+\int_0^t\fr{1}{T-s}dG_s(F_s).\label{eqn:DR:VIXinG}
\end{equation}
By the Ito-Ventzel formula (Lemma~\ref{lem:VIXSDE:ItoV}),
\begin{align*}
     G_t(F_t) 
            ={}&G_0(F_0)+\int_0^t\int_0^{F_s}\fr{d\hat{P}(k,T,s,\vX_s)}{k^2}\,m(dk) \\{}&+\int_0^t\int_{F_s}^\infty\fr{d\hat{C}(k,T,s,\vX_s)}{k^2}\,m(dk)
+\int_0^t\fr{\pr G_s}{\pr f}(F_s)\,dF_s\\
            {}&+\sum_{j=0}^d\int_0^t\fr{1}{F_s^2}\left(\fr{\pr\hat{P}}{\pr x_j}(F_s,T,s,\vX_s)-\fr{\pr\hat{C}}{\pr x_j}(F_s,T,s,\vX_s)\right) d\qv{X^j}{F}_s\\
            {}& +\fr{1}{2}\int_0^t\fr{\pr^2 G_s}{\pr f^2}(F_s)\,d\qv{F}{F}_s.
\end{align*}
Therefore
\begin{align*}
     G_t(F_t) 
            ={}&G_0(F_0)+\int_0^t\int_0^{F_s}\fr{d\hat{P}(k,T,s,\vX_s)}{k^2}m(dk) +\int_0^t\int_{F_s}^\infty\fr{d\hat{C}(k,T,s,\vX_s)}{k^2}m(dk)  \\
            {}&+\int_0^t\fr{1}{F_t^2}\left(\hat{P}(F_s,T,s,\vX_s)- \hat{C}(F_s,T,s,\vX_s)\right)\,dF_s\\
            {}&+\sum_{j=0}^d\int_0^t\fr{1}{F_t^2}\left(\fr{\pr\hat{P}}{\pr x_j}(F_s,T,s,\vX_s)-\fr{\pr\hat{C}}{\pr x_j}(F_s,T,s,\vX_s)\right)\,d\qv{X^j}{F}_s\\
            {}&+\fr{1}{2}\int_0^t\left.\fr{\pr}{\pr f}\left(\fr{\hat{P}(f,T,s,\vX_s)-\hat{C}(f,T,s,\vX_s)}{f^2}\right)\right|_{f=F_s}d\qv{F}{F}_s.
\end{align*}
To simplify further, observe that
\begin{equation*}
     {\hat{P}(F_t,T,t,\vX_t)}-{\hat{C}_t(F_t,T,t,\vX_t)} = 0,
\end{equation*}
\begin{align*}
    \sum_{j=0}^d\fr{1}{F_t^2}\left(\fr{\pr\hat{P}}{\pr x_j}(F_t,T,t,\vX_t)\right.&\left.-\fr{\pr\hat{C}}{\pr x_j}(F_t,T,t,\vX_t)\right) d\qv{X^j}{F}_t\\
    &=\fr{1}{F_t^2}\left.\fr{\pr}{\pr x_0}\left(F_t-x_0\right)\right|_{x_0=F_t} d\qv{F}{F}_t\\
         &=-\fr{1}{F_t^2}d\qv{F}{F}_t,
\end{align*}
and
\begin{align*}
    \left.\fr{\pr}{\pr f}\left(\fr{\hat P(f,T,t,\vX_t)-\hat C(f,T,t,\vX_t)}{f^2}\right)\right|_{f=F_t}
        ={}& -\fr{2}{F_t^3}\left.\left(\hat P(f,T,t,\vX_t)-\hat C(f,T,t,\vX_t)\right)\right|_{f=F_t}\\
        {}&+\fr{1}{F_t^2}\fr{\pr}{\pr f}\left.\left(\hat P(f,T,t,\vX_t)-\hat C(f,T,t,\vX_t)\right)\right|_{f=F_t}\\
        ={}&\fr{1}{F_t^2}\left.\fr{\pr}{\pr f}\left(f-F_t\right)\right|_{f=F_t}\\
        ={}&\fr{1}{F_t^2},
\end{align*}
for all $F_t$ by put-call parity.  Finally,
\begin{align}
     G_t(F_t)
     ={}&G_0(F_0)+\int_0^t\int_0^{F_s}\fr{d\hat{P}(k,T,s,\vX_s)}{k^2}m(dk) \nonumber\\ {}&+\int_0^t\int_{F_s}^\infty\fr{d\hat{C}(k,T,s,\vX_s)}{k^2}m(dk) -\fr{1}{2}\int_0^t\fr{1}{F_s^2}d\qv{F}{F}_s \nonumber\\
            ={}&G_0(F_0)+\int_0^t \int_{0}^\infty\fr{1}{k^2}{d\hat\Theta}(k,T,s,\vX_s)\,m(dk)
                -\fr{1}{2}\int_0^t\fr{1}{F_s^2}d\qv{F}{F}_s\nonumber\\
            ={}&G_0(F_0)+\int_0^t\int_0^\infty\fr{1}{k^2}\fr{\pr\hat \Theta}{\pr t}(k,T,s,\vX_s)\,m(dk)\,ds
            +\sum_{i=0}^d\int_0^t\int_0^\infty\fr{1}{k^2}\fr{\pr\hat \Theta}{\pr x_i}(k,T,s,\vX_s)\,m(dk)\,dX^i_s\nonumber\\
                {}&+\fr{1}{2}\sum_{i,j=0}^d\int_0^t\int_0^\infty\fr{1}{k^2}\fr{\pr^2\hat \Theta}{\pr x_i\pr x_j}(k,T,s,\vX_s)\,m(dk)\,d\qv{X^j}{X^i}_s-\fr{1}{2}\int_0^t\fr{1}{F_s^2}d\qv{F}{F}_s.\label{eqn:DR:GSIE}
\end{align}
%
Combining Equation~\eqref{eqn:DR:VIXinG} and Equation~\eqref{eqn:DR:GSIE} completes the proof.  The stochastic differential equation version of the result is stated in Equation~\eqref{eqn:DR:wtVIXsquared}.
\end{proof}

\section{Auxiliary Results}
\label{app:AL}
\begin{lem}(\cite{BaRu10})
\label{lem:Defn:BR}
    Let $f:\R_+\to\R$ be twice differentiable almost everywhere with respect to the Lebesgue measure, so that
\begin{equation*}
        \int_{m_1}^{m_2}|f''(k)|dk<+\infty,\,\,\forall m_1,\,\,m_2\in\R_+, \mbox{ s.t. } m_2>m_1>0,
\end{equation*}
and fix $y\in\R_+$. Then for any $x \in\R$,
\begin{equation*}
    \int_y^\infty(x-k)^+|f''(k)|dk<+\infty,
\end{equation*}
\begin{equation*}
    \int_0^y(k-x)^+|f''(k)|dk<+\infty
\end{equation*}
and
\begin{equation*}
    f(x)=f(y)+f'(y)(x-y)+\int_y^\infty(x-k)^+f''(k)dk+\int_0^y(k-x)^+f''(k)\,dk.
\end{equation*}
\end{lem}

\begin{lem}(Ito-Ventsel Formula)
    \label{lem:VIXSDE:ItoV}
    Let $G_t(x)$ be a family of stochastic processes, continuous in $(t,x)\in(\R^+\times\R^d)$ $\mP$-a.s. satisfying:
    \begin{enumerate}[(i)]
        \item For each $t>0$, $x\to G_t(x)$ is $\mathcal{C}^2$ from $\R^d$ to $\R$.
        \item For each $x$, $(G_t(x),t\ge 0)$ is a continuous semi-martingale
        \begin{equation*}
                dG_t(x) = \sum_{j=1}^n g_t^j(x)dM_t^j
        \end{equation*}
        where $M^j$ are continuous semi-martingales, and $g^j(x)$ are stochastic processes continuous in $(t,x)$, such that $\forall s>0$, $x\to g^j_s(x)$ are $\mathcal{C}^2$ maps, and $\forall x$, $g^j(x)$ are adapted processes.
    \end{enumerate}
    Let $X = (X^1,...,X^d)$ be a continuous semi-martingale. Then
    \begin{align*}
        G_t(X_t) ={}& G_0(X_0) +\sum_{j=1}^n\int_0^tg^j_s(X_s)dM_s^j
                    +\sum_{i=1}^d\int_0^t\fr{\pr G_s}{\pr x_i}(X_s)dX^i_s\\
            {}&+\sum_{i=1}^d\sum_{j=1}^n\int_0^t \fr{\pr g^j_s}{\pr x_i}(X_s)d\qv{M^j}{X^i}_s+\fr{1}{2}\sum_{i,k=1}^d\int_0^t\fr{\pr^2 G_s}{\pr x_i \pr x_k}(X_s)d\qv{X^k}{X^i}_s.
    \end{align*}
\begin{proof}
    The version stated above is taken from \cite{JeYC09}. For the original result see \cite{Vent65}.
\end{proof}
\end{lem}

\bibliographystyle{plainnat}
\small{\bibliography{references}}

\end{document}